\documentclass[letterpaper]{article} 
\usepackage{aaai25}  
\usepackage{times}  
\usepackage{helvet}  
\usepackage{courier}  
\usepackage[hyphens]{url}  
\usepackage{graphicx} 
\usepackage{natbib}  
\usepackage{caption} 
\usepackage[colorinlistoftodos]{todonotes}
\usepackage{tablefootnote,multirow}
\usepackage{indentfirst, dirtytalk}
\usepackage{mathtools}
\usepackage{booktabs}
\usepackage{amsthm, amsfonts}
\usepackage{theoremref}
\usepackage{thmtools, thm-restate}
\usepackage{mathrsfs}
\usepackage[capitalise]{cleveref}
\usepackage{sgame,array}
\usepackage{tabularx}
\usepackage[para,online,flushleft]{threeparttable}
\usepackage{algorithm}
\usepackage[noend]{algpseudocode}
\usepackage{algorithmicx}
\usepackage{caption}
\usepackage{subcaption}
\usepackage{tikz}
\usetikzlibrary{tikzmark}

\usepackage{pifont}
\usepackage{pgfplots}
\pgfplotsset{compat=1.18}
\usepackage{laura}

\title{Fixed-budget and Multiple-issue Quadratic Voting}
\author {
    Laura Georgescu\textsuperscript{\rm 1},
    James Fox\textsuperscript{\rm 1},
    Anna Gautier\textsuperscript{\rm 2},
    Michael Wooldridge\textsuperscript{\rm 1}
}
\affiliations {
    \textsuperscript{\rm 1}University of Oxford\\
    \textsuperscript{\rm 2}KTH Royal Institute of Technology\\
    georgescu.i.laura@gmail.com, james.fox@cs.ox.ac.uk, annagau@kth.se, mjw@cs.ox.ac.uk
}

\begin{document}
\maketitle
\begin{abstract}
\emph{Quadratic Voting} (QV) is a social choice mechanism that addresses the \say{tyranny of the majority} of one-person-one-vote mechanisms. Agents express not only their preference ordering but also their preference intensity by purchasing $x$ votes at a cost of $x^2$. Although this pricing rule maximizes utilitarian social welfare and is robust against strategic manipulation, it has not yet found many real-life applications. One key reason is that the original QV mechanism does not limit voter budgets. 
Two variations have since been proposed: a (no-budget) multiple-issue generalization and a fixed-budget version that allocates a constant number of \say{credits} to agents for use in multiple binary elections. While some analysis has been undertaken with respect to the multiple-issue variation, the fixed-budget version has not yet been rigorously studied.
In this work, we formally propose a novel fixed-budget multiple-issue QV mechanism. This integrates the advantages of both the aforementioned variations, laying the theoretical foundations for practical use cases of QV, such as multi-agent resource allocation. We analyse our fixed-budget multiple-issue QV by comparing it with traditional voting systems, exploring potential
collusion strategies, and showing that checking whether strategy profiles form a Nash equilibrium is tractable.
\end{abstract}

%

\section{Introduction}

In an attempt to alleviate some of the pitfalls prevalent in the widely used one-person-one-vote (1p1v) voting system, economists, philosophers and mathematicians have proposed a variety of alternative electoral mechanisms. In \emph{On Liberty} \cite{on_liberty}, John Stuart Mill investigated decision-making from the stance of utilitarianism and argued that democratic ideals might lead to the oppression of minorities, referred to as the \say{tyranny of the majority.} Inspired by existing market mechanisms for allocating private goods, \citet{qv2018} introduced \emph{Quadratic Voting} (QV) as an alternative collective decision-making mechanism for public goods that, in contrast to 1p1v, allows potentially Pareto improving trades (that is, benefiting all involved parties), while still being a form of direct democracy. The key idea of QV is that for every additional unit of influence, agents pay a price linear to the magnitude of total impact. Therefore, the overall price paid is quadratic in the total number of votes and thus electoral power per monetary unit is inversely proportional to the total influence.

\vspace*{1.5ex}\noindent\textbf{Related work:}
\citet{qv2018} showed that, assuming that agents agree on the probability that an extra vote will change the outcome, a pricing rule is optimal (i.e., in the equilibrium, the outcome coincides with the optimal utilitarian social welfare choice) if and only if it is quadratic. \citet{robustness} showed that the mechanism is somewhat robust against certain forms of collusion, fraud, and aggregate uncertainty. The original QV proposal focused on binary decisions where payments are real money, but \citet{multiplealternative} extended QV to multiple-issue elections and \citet{quad_election_law} tried to eliminate the influence of wealth by allocating the same fixed number (or budget) of \say{credits} to agents for them to divide between several binary elections.

QV has been examined from the ethical~\cite{ethical_condierations}, practical \cite{publicgood, inthewild} and legal \cite{quad_election_law} perspectives, with applications including corporate \cite{coportategovernance} and blockchain \cite{blockchain} governance, civic engagement \cite{bassetti2023civicbase}, and survey research \cite{inthewild}. Still, there lacks a rigorous comparison of QV against other voting mechanisms, and a formalization of a fixed-budget multiple-issue variation is absent, even though this is arguably the most likely real-world implementation.

\vspace*{1.5ex}\noindent\textbf{Our Contributions:}
Our original contributions are threefold. First, in \cref{sec:fixed_budget}, we propose the first formalization of fixed-budget multiple-issue QV, which inherits the advantages of both no-budget multiple-issue QV \cite{multiplealternative} and fixed-budget QV (described informally in~\cite{quad_election_law}). Second, in \cref{sec:properties}, we rigorously compare QV with other more established voting methods and explore novel collusion strategies. Finally, in \cref{sec:complexity}, we prove that checking whether a given strategy profile is a pure Nash equilibrium (NE) for QV can be done in polynomial time for both fixed-budget and no-budget multiple-issue QV.

\section{Preliminaries}
In this section, we formally introduce previously defined variants of QV and summarise existing known results. \Cref{tbl:variants} categorizes the QV versions based on the number of outcomes and on whether there is a credit-imposed spending limit for purchasing the votes.

\begin{table}[t]
	\centering
	\caption{Variants of QV}
	\begin{threeparttable}
	\begin{tabularx}{\linewidth}{ c X X } 
		\toprule
		& No budget & Fixed budget \newline  (credits) \\ 
		\midrule
		\multirow{2}{4em}{Binary decision} & QV \newline \citep{qv2018}  & Fixed-budget QV \newline \cite{quad_election_law} \\
		\addlinespace
		\multirow{2}{4em}{Multiple outcomes} & Multiple-issue QV \cite{multiplealternative}  & Fixed-budget Multiple-issue QV\tnote{1} \\ 
		\bottomrule%
	\end{tabularx}
	\begin{tablenotes}
		\item[1] Our contribution, formalized in \cref{sec:fixed_budget}.
	\end{tablenotes}
	\end{threeparttable}
	\label{tbl:variants}
\end{table}%

First, we introduce some notation. The \textit{utility} $\util{i}{\omega} \in \R$ captures the preference of
agent $i$ towards outcome $\omega$; $\Sigma$ is the set of
available votes (strategies) for every agent, and a \textit{strategy
	profile} $\strtprofile \in \Sigma^{|N|}$ denotes a tuple $(\allvotes{1}, \dots, \allvotes{|N|})$ containing a strategy for each of the agents in the society $N$, respectively. $U^i(\strtprofile)$
is also a utility, which represents the preference of agent $i$
towards the result of the votes $\strtprofile$.  The
\textit{(utilitarian) social welfare} is
simply the sum of every agent's utility $\sum_i U^i(\strtprofile)$. Agent $i$'s strategy $\sigma \in \Sigma$ is a \textit{best response} to the partial strategy profile $\votes{-i}{}$ assigning strategies to the other agents if $U^i(\votes{-i}{}, \sigma) \geq U^i(\votes{-i}{}, \sigma'), \forall \sigma' \in \Sigma$. A strategy profile $\strtprofile$ is a \textit{(pure) Nash equilibrium} (NE) if no agent has an incentive to unilaterally change their strategy, meaning that every agent is simultaneously playing a best response: $ U^i(\strtprofile) \geq U^i(\votes{-i}{}, \sigma), \forall i \in N, \sigma \in \Sigma$.



\newcommand{\sgn}[0]{\text{sgn}}
\newcommand{\refundterm}[0]{\frac{\alpha}{N - 1} \sum_{j \neq i} (\allvotes{j})^2}



\subsection{No-budget binary and multiple-issue QV}
\label{prelim:multipleqv}
\citet{qv2018} introduced binary QV as a mechanism to avoid the \say{tyranny of the majority} (disregard towards minorities' preferences) in a binary collective decision, where a society of agents $N$ has to decide whether or not to adopt a motion~$\mathcal{M}$. \emph{Multiple-issue QV} \cite{multiplealternative} generalizes binary QV by allowing agents to decide between an arbitrary number of issues (or outcomes) $\Omega$.

In both systems, for a system-defined constant $\alpha > 0$, each agent $i$ incurs a price of $\alpha (\votes{i}{\omega})^2$ for every $\votes{i}{\omega}$ votes placed on issue $\omega$. The number of votes can be either positive, representing support, or negative for opposition. For example, if there are three issues and an agent casts $4$ votes for the first, $-3$ for the second, and $1$ for the last, it costs them $\alpha(4^2+3^2+1^2)=26\alpha$.  Each agent can cast unlimited votes, and the winning outcome is selected uniformly at random from $W$, the set of outcomes with the highest total vote count $\totalvotes{\omega}$, or, in the binary paradigm, agents cast positive or negative votes on a single motion, which will be adopted if the total number of votes is non-negative (i.e., if $\totalvotes{\mathcal{M}} = \sum_{i \in N} \allvotes{i} \geq 0$). The voting mechanism is budget-balanced, meaning that the profits made from one agent are redistributed to the other agents, each one receiving $\alpha (\votes{i}{\omega})^2/(|N|-1)$ from agent $i$. We now formalize this intuition:

\begin{definition}[No-budget multiple-issue QV] 
	\label{def:nobudgetQV}
	A no-budget multiple-issue QV election mechanism has the following components:
	\begin{enumerate}
		\item $N$ is the set of agents voting in the election.
		\item $\Omega$ is the finite set of outcomes.
		\item $u \in \R^{|N| \times |\Omega|}$ is the utility matrix, where $\util{i}{\omega}$ is the utility outcome $\omega$ brings to agent $i$ if $\omega$ is elected as a winner. The vectors $u^1,\dots, u^{|N|}$ are i.i.d. In multiple-issue QV, it is assumed that $u$ is common knowledge among the society $N$, whereas in binary QV the agents just know that their utility values are drawn from the same distribution.
		
		\item The strategy of agent $i$ is denoted by $\allvotes{i} \in \Sigma = \Z^{|\Omega|}$, where $\votes{i}{\omega}$ represents the (integer) number of votes agent $i$ puts for outcome $\omega$.
		$\strtprofile \in \Z^{|N| \times |\Omega|}$ is a strategy profile. 
		\item $\alpha \in \R_{++}$ is the constant that controls the magnitude of the payment. For casting ballot $\allvotes{i}$, agent $i$ pays a total price of $\alpha \sum_{\omega} (\votes{i}{\omega})^2$ and is refunded the mean of the prices paid by the other agents $\frac{\alpha}{|N| - 1} \sum_{j \neq i} \sum_{\omega \in \Omega} (\votes{j}{\omega})^2$.
		
		\item The total number of votes for outcome $\omega$ is $\totalvotes{\omega} = \sum_{i \in N} \votes{i}{\omega}$. 
		
		\item The set of possible winning outcomes, $W$, is the set the outcomes with the highest support $W = \argmax_{\omega \in \Omega} \totalvotes{\omega}$. The winning candidate is selected uniformly at random from $W$, implying that the probability outcome $\omega \in \Omega$ is elected is: 
		$$\P{\omega} =
		\begin{cases}
			1/|W|, & \text{ if } \omega \in W \\
			0, & \text{ otherwise}.
		\end{cases}$$
		
		\item The (total) utility $U^i(\strtprofile)$ for agent $i \in N$, in terms of the strategy profile $\strtprofile$, is their expected utility minus the overall monetary deficit. Formally, $U^i(\strtprofile) ={}$
		\begin{equation*} \small
			\label{eq:indiv_util}
			\underbrace{\sum_{\omega \in \Omega} \util{i}{\omega}\P{\omega} }_{\text{expected outcome}}%
			- \underbrace{\alpha \sum_{\omega \in \Omega} (\votes{i}{\omega})^2}_{\text{payment}}%
			+ \underbrace{\frac{\alpha}{N - 1} \sum_{j \neq i} \sum_{\omega \in \Omega} (\votes{j}{\omega})^2}_{\text{refund}}.
		\end{equation*}
	\end{enumerate}
	
	Note that when the agent is optimizing for outcome utility, the refund term in $U^i(\strtprofile)$ can be ignored as it is independent from $i$'s strategy. A no-budget binary QV election is a special case of this mechanism, where the motion $\mathcal{M}$ is adopted if and only if $\totalvotes{\mathcal{M}} \geq 0$ and since no probabilities are involved in deciding the outcome of the election, the outcome term in (8)'s total utility is replaced by $u^i \sgn(\totalvotes{\mathcal{M}})$. Also, note that agents only express one number of votes in this case, positive for the motion, negative means against.
\end{definition}
%
%
\begin{example}
	Given the ballots in \cref{tbl:multipleqv} (top), then outcome $\omega_1$ gets $\totalvotes{\omega_1}=6-4+1=3$ votes whilst $\omega_2$ and $\omega_3$ get  $\totalvotes{\omega_2}=3$ and $\totalvotes{\omega_3}=-2$ votes, respectively. Both $\omega_1$ and $\omega_2$ tie for the winning position, so $W = \{\omega_1, \omega_2\}$ and $\P{\omega_1} = \P{\omega_2} = 0.5$. For $\alpha=1$, the prices incurred by every agent are presented in \cref{tbl:multipleqv} (bottom).
	
	\begin{table}
		\centering
		\caption{Multiple-issue QV example. Purchased votes $\strtprofile$ (top) and cash flow for $\strtprofile$ (bottom). }
		\label{tbl:multipleqv}
		\begin{subtable}[b]{0.3\textwidth}
			\centering
			\begin{tabular}{c c c c}
				\toprule
				$\Omega$ & A & B & C \\
				\midrule
				$\omega_1$ & 6 & -4 & 1 \\
				$\omega_2$ & -3 & 5 & 1 \\
				$\omega_3$ & 1 & -10 & 7\\
				\bottomrule
			\end{tabular}
		\end{subtable}\quad
		\begin{subtable}[b]{0.5\textwidth}
			\centering
			\begin{tabular}{c c c c}
				\toprule
				& A & B & C \\
				\midrule
				\multirow{2}{*}{payment} & {\footnotesize$6^2+3^2+1^2$} & {\footnotesize$4^2+5^2+10^2$} & {\footnotesize$1^2+1^2+7^2$} \\
				& 46 & 141 & 51 \\
				\addlinespace
				refund & $\vphantom{\frac{\mathstrut 1}{\mathstrut 2}} \frac{141+51}{2}$ & $\frac{46+51}{2}$ & $\frac{46+141}{2}$\\
				\bottomrule
			\end{tabular}
		\end{subtable}
	\end{table}
\end{example}

\subsection{Fixed-budget QV}

To investigate the viability of introducing QV as a political voting system, \citet{quad_election_law} suggest fixed-budget QV (which they call \textit{mQV}) to overcome the problematic use of real money. Fixed-budget QV allocates all agents the same number of credits or tokens $B$, where every credit replaces $1$ monetary unit in the traditional no-budget QV. After every agent receives their budget of $B$ tokens, they must decide how to split them across a number of binary referenda, which may happen once every $5$ years. Observe that when $B=1$ and there is one election, fixed-budget QV is equivalent to 1p1v where every agent picks which election they are interested in (with the extension that one can also place $-1$ votes).

Hence, fixed-budget QV aims to eliminate the influence of wealth in elections and allows minorities to have more proportional power by being more resistant towards the \say{tyranny of the majority}. Both no-budget QV and the fixed-budget variation use the same quadratic payment rule (with money or credits, respectively) and the same system for determining the winning outcome, with the goal of maximising social welfare. However, the existing literature presents no formalization or rigorous analysis of this variant.

\section{Fixed-budget multiple-issue QV}
\label{sec:fixed_budget}

We now formally introduce \emph{fixed-budget multiple-issue QV}, for a single, multiple-issue, election, which combines the advantages of both multiple-issue QV \cite{multiplealternative} and fixed-budget QV \cite{quad_election_law}. We believe that this variation presents the most practical option for real-world usage.



\begin{definition}[Fixed-budget multiple-issue QV]
	\label{def:fixedbudgetQV}
	A fixed-budget multiple-issue QV election mechanism is formalised with (1)-(7) the same as in \cref{def:nobudgetQV} and:
	\begin{enumerate}
		\item A budget $B \in \mathbb{N}$, for each agent to spend on the election.
		\item Each agent $i \in N$ spends at most $B$ credits, i.e., 
		\begin{equation*}
			\sum_{\omega \in \Omega}(\votes{i}{\omega})^2 \leq B, \quad \forall i \in N.
		\end{equation*}
		\item Agent $i$'s total utility is the expected outcome utility, i.e., \\ $U^i(\strtprofile) = \sum_{\omega \in \Omega} \util{i}{\omega}\P{\omega}$. The payment and refund terms disappeared from (8) in \cref{def:nobudgetQV} as all agents receive the same number of credits. Note that the payment scalar $\alpha$ is not needed, as it can be incorporated in $B$.

		
	\end{enumerate}

\end{definition}

	
	

First, we wish to ask how the properties of fixed-budget multiple-issue QV compare with the no-budget case. On the one hand, \citet{multiplealternative} argued in their no-budget multiple-issue QV analysis that as the population size grows, the agents' monetary involvement diminishes, and so, however low the budget $B > 0$, a sufficiently large society can rely on the participation of all agents. This implies that the no-budget multiple-issue QV results will asymptotically transfer to the fixed budget case.
On the other hand, even though no-budget multiple-issue QV is safe from the tyranny of the majority, the fixed budget eliminates the influence of personal capital on the election. This makes it even more resistant towards the tyranny of the majority. We now prove that other important properties of multiple-issue QV remain the same with or without budget.



\newcommand{\titleofProp}[0]{Properties: QV vs other voting systems}
\subsection{\titleofProp} 
\label{sec:properties}
QV claims to offer many advantages as a voting system \cite{qv2018, robustness, ethical_condierations,quad_election_law}, but it has not yet been rigorously compared against other more established voting systems. So, in this section, we compare both fixed-budget and no-budget multiple-issue QV with four popular alternatives.
Meaningful properties of voting mechanisms relate to characteristics of the elected candidate, how the winner changes in related circumstances, and whether voters may have an incentive for strategic manipulation. In particular, we evaluate QV with respect to:

\begin{description}
	\item[Intensity:] A voting system satisfies this if an agent can express the intensity of their preferences, not just the ordering.
	\item[Majority safe:] A voting system satisfies this when a candidate that is ranked first by at least half of the agents is not guaranteed to win the election.
	\item[Consistency:] A voting method is \textit{consistent} if the election comprised of the sum of several elections all giving outcome $\omega$ also gives the same winning outcome.
	\item[Clone independence:]
	The winner of an election must not change after a non-winning clone $\gamma$ of an existing candidate $\omega$, was introduced to the ballot. By clone we refer to an outcome which is very similarly preferred as $\omega$ by all agents. That is, the utility of the clone is $\util{i}{\gamma} = \util{i}{\omega} - \epsilon, \forall i$ for a very small $\epsilon > 0$. We assume that strategic manipulation (i.e., changing the votes for the other outcomes) is not allowed, as in \citet{cloneproof, range_ICC_AFB}.
	\item[Independent of irrelevant
	alternatives (IIA):] A voting system satisfies this if from the set of issues $\{\omega,\phi\}$ the
	overall preferred outcome is $\omega$, then, when irrelevant alternative
	$\psi$ is introduced and only votes for $\psi$ change, $\omega$ must
	still be preferred over $\phi$.
	\item[No favourite betrayal (NFB):] A decision mechanism passes the NFB criterion if there is no incentive for an agent to rate another outcome strictly ahead of their sincere favorite, in order to make a more preferred outcome win.
\end{description}
\cref{th:properties} summarizes results for QV for all five criteria. The proof can be found in \cref{appendix:properties}.

\begin{restatable}{theorem}{properties}
	Multiple issue QV (fixed budget or not) is safe from the tyranny of the majority, it is consistent, clone-independent, IIA, but it does not satisfy NFB.
	\label{th:properties}
\end{restatable}


We compare QV against the following voting systems whose properties according to the criteria above are known \cite{axiomatic_char, cloneproof, 1p1vFB, iia1}:

\begin{description}
	\item[1p1v] \cite{intro_to_vote_counting} (otherwise known as first-past-the-post or plurality) is the most widely-used voting system, in which agents are allowed to vote for at most one candidate, and the candidate with the most number of votes wins.
	
	\item[Borda] \cite{intro_to_vote_counting}. This requires agents to rank all the candidates and for each ballot, the first option gets $|\Omega|-1$ points, the second gets $|\Omega| - 2$, $\dots$, until the last option gets $0$. The winner is the outcome with the highest sum of such points.
	
	\item[Approval] \cite{intro_to_vote_counting}, which has agents indicate support or not for each of the candidates. The winner is the outcome supported by the most number of agents.
	
	\item[Score voting]  (or range voting),  is similar to approval voting, but agents are allowed to express their support for each candidate on a scale from $0$ to $k$ points, where $k$ is usually $5$ or $10$. The candidate with the most points wins\cite{range_ICC_AFB}.
\end{description}


QV is similar to score voting, in the sense that agents can express their preference intensity by casting more or fewer votes. The difference comes from the levels of support that can be shown, which is quite strict for score voting and much more diverse for QV. Additionally, score voting uses \say{free} votes, whereas in QV agents must purchase votes using money or credits, therefore inhibiting agents from expressing extreme preferences.
%
%
\newcommand{\cmark}{\textcolor{teal}{\ding{51}}}%
\newcommand{\xmark}{\textcolor{red}{\ding{55}}}%
\newcommand{\mmark}{\textcolor{olive}{\ding{70}}}
%
%
\begin{table*}[t]
	\small
	\centering
	\caption{Comparison of QV against other voting systems.\\ \cmark\hspace{0.5ex}means satisfying the property;  \xmark\hspace{0.5ex}means not satisfying the property; \mmark\hspace{0.5ex}represents limited satisfaction, see the footnote.}
	\begin{threeparttable}
	\begin{tabularx}{\linewidth}{c c X X X X}
		\toprule
		& QV & Approval & Score & 1p1v & Borda \\
		\midrule
		Intensity & \cmark & \xmark & \mmark\tnote{2} & \xmark & \xmark \\
		Majority safe & \cmark & \mmark\tnote{3} & \cmark\tnote{4} & \xmark \cite{majority} & \cmark\tnote{4} \\
		Consistency & \cmark & \cmark \cite{axiomatic_char} & \cmark \cite{axiomatic_char} & \cmark \cite{axiomatic_char} & \cmark \cite{axiomatic_char} \\
		Clone indep & \cmark & \cmark \cite{cloneproof} & \cmark \cite{range_ICC_AFB} & \xmark \cite{cloneproof} & \xmark \cite{cloneproof}\\
		IIA & \cmark \tnote{5} & \cmark\tnote{5} \cite{iia1} & \cmark\tnote{5} \cite{iia1} & \xmark \cite{iia1} & \xmark \cite{iia1} \\
		NFB & \xmark & \cmark \cite{fav_bet} & \cmark \cite{range_ICC_AFB} & \xmark \cite{1p1vFB} & \xmark \cite{1p1vFB} \\ 
		\bottomrule
	\end{tabularx}
	\begin{tablenotes}
		\item[2]Limited to the maximum score on the ballot. Usually, the maximum score $k$ does not allow for high granularity. It is common that $k=5$ or $k=10$.
		\item[3]Depends on how we interpret preference; \cite{the_case_for_AV} contains a discussion about the different models and how they affect the criterion.
		\item[4] Due to expression of intensity.\label{no_majority}
		\item[5] Assuming that agents evaluate issues independently, using their own scale. \label{indiv_IIA}
	\end{tablenotes}
	\end{threeparttable}
	\label{tbl:big_table}
\end{table*}

We summarise existing results on the aforementioned voting mechanisms against our results for QV in \Cref{tbl:big_table}. There are several other known advantages of QV compared with other voting mechanisms. First, under the assumption that the utilities of the agents $u^1, u^2, \dots, u^{|N|}$ are independent and identically distributed,  it is shown in \citet{robustness} that QV is the simplest mechanism  with outcomes that optimize social welfare. Additionally, \citet{qv2018} show that when the population agrees on the the probability $p$ that an additional vote will change the outcome of the election, the quadratic formula is the unique price-taking rule that accomplishes utilitarian optimality, i.e., at the equilibrium $\neprofile$, the sum of the votes $\totalvotes{\omega}$ agrees with the true preference of the whole society. \citet{multiplealternative} prove that this is also asymptotically\footnote[6]{In a sequence of infinite games, where agents vote using a pure Nash equilibria strategy. The variable in these games is the realization of $u$.} true in the multiple-issue QV case.


Another advantage of no-budget QV is that the agents are unrestrained on the number of votes they can cast as long as they are willing to give up the necessary funds. By allowing expression of high intensity of preferences, QV protects against the tyranny of the majority, as argued in \cref{th:properties}. However, a plausible concern is that wealthy individuals have undue influence over the election. The quadratic rule mitigates this by reducing the effectiveness of monetary units for those who cast more votes. To address this further, the fixed-budget QV variant completely removes the wealth difference between the agents by allocating them the same number of credits. \citet{qv2018} describe QV as \say{an optimal intermediate point between the extremes of dictatorship and majority rule.}

Disadvantages of no-budget QV include favourite betrayal (\cref{th:properties}) and that some of the agents might prefer to abstain from participating. \citet{multiplealternative} demonstrate that as $|N|$ tends to infinity, the monetary contribution of all agents decreases. Still, there are cases in which agents might obtain higher utility by abstaining: imagine that one (almost) indifferent agent $A$ has information that $B$ has very strong preferences, then, since changing the elected issue is very expensive, $A$ will not cast any votes (\cref{tbl:abstain}). Nevertheless, the most concerning disadvantage of QV, both budget-limited or not, is collusion, discussed next.

\begin{table}[t]
	\centering
	\caption{Abstaining from voting. If agent $A$ tried to make $X$ win or tie with $Y$, its final utility would be lower than the current strategy of abstaining.}
	\begin{tabular}{c c c c c}
		\toprule
		& $u^A$ & $u^B$ & $\allvotes{A}$ & $\allvotes{B}$ \\
		\midrule
		$X$ & 40 & 0 & 0 & -5 \\
		$Y$ & 30 & 400 & 0 & 5 \\
		\bottomrule
	\end{tabular} 
	\label{tbl:abstain}
\end{table}

\newcommand{\titleofColl}[0]{Collusion}
\subsection{\titleofColl}
\label{sec:collusion}
\newcommand{\xvotes}[2]{x^{#1}_{#2}}

We now examine the susceptibility of QV to collusion, which occurs when agents collaborate to obtain lower vote prices; if votes are spread among multiple individuals, then the marginal cost of buying one more vote can become significantly smaller. We begin by summarising \citet{robustness}'s analysis of collusion before proposing a new model which introduces additional restrictions based on rational behaviour assumptions. Throughout this subsection, the set of colluding agents $\mathscr{C}$ update their ballots from $v^{\mathscr{C}}(u)$ to the post-manipulation votes $x^{\mathscr{C}}$, while not changing the total number of votes placed for any outcome. 

\citet{robustness}'s analysis of strategic manipulation assumes that the coalition $\mathscr{C}$ is trying to maximise the joint utility of its members. To minimize the sum of payments incurred by the agents, the optimal strategy is for every agent in $\mathscr{C}$ to cast as many votes for outcome $\omega$ as the mean number of votes all coalition members placed for $\omega$ (as a result of the quadratic price). However, this analysis assumes that agents are selfless and benevolent, seeking the maximal well-being of the whole society. A more reasonable assumption, common in cooperative game theory \cite{chalkiadakis2022computational}, is that agents are, first and foremost, individualists and they would not participate in the coalition if their payment increases. We therefore characterise a stronger definition of collusion: the total sum of votes per outcome remains the same, but (i) no individual is made worse-off as part of the coalition and (ii) there is at least one agent who strictly benefits. This is formally defined in \cref{def:collusion}.

\begin{definition}(Strictly-beneficial collusion)
	\label{def:collusion}
	The set of agents $\mathscr{C}$ would benefit from colluding if and only if $\exists x\in \Z^{|\mathscr{C}|x|\Omega|}$ redistribution of their votes such that 
	
	\begin{enumerate}
		\item The total sum of votes per outcome remains the same:
		\begin{equation}
			\label{eq:totalconstant}
			\sum_{i \in \mathscr{C}} \votes{i}{\omega} = \sum_{i \in \mathscr{C}} \xvotes{i}{\omega}, \forall \omega \in \Omega.
		\end{equation}
		
		\item It satisfies individual rationality - all agents end up paying no more than if they were to vote independently:
		\begin{equation}
			\label{eq:indiv_rational}
			\forall i \in \mathscr{C} : \sum_{\omega \in \Omega} (\votes{i}{\omega})^2 \geq \sum_{\omega \in \Omega} (\xvotes{i}{\omega})^2.
		\end{equation}
		
		\item The price reduces for at least one member $c_+$ of the coalition:
		\begin{equation}
			\exists c_+ \in \mathscr{C} : \sum_{\omega \in \Omega} (\votes{c_+}{\omega})^2 > \sum_{\omega \in \Omega} (\xvotes{c_+}{\omega})^2.
			\label{eq:strict}
		\end{equation}
	\end{enumerate}
\end{definition}
We now outline two methods for colluding.
\noindent
\paragraph{Method 1: When an outcome has both pro and contra votes}
Consider an outcome $\omega$ which has both positive and negative votes; the votes for all other outcomes will remain unchanged by the coalition. We present a way to reallocate the votes for $\omega$ by eliminating the weaker direction of votes, such that all agents end up paying less than (or equal to) their original payments.

Since we are considering only one outcome for redistribution, we simplify
the notation such that $\allvotes{i} = \votes{i}{\omega}$ and
$x_i = \xvotes{i}{\omega}$. Let
$s^+(\strtprofile) = \sum_{i \in \mathscr{C} \land \allvotes{i} > 0}
\allvotes{i}$ be the sum of the positive votes and
$s^-(\strtprofile) = \sum_{i \in \mathscr{C} \land \allvotes{i} < 0 }
-\allvotes{i}$ be the absolute value of the negative support. Both
$s^+(\strtprofile)$ and $s^-(\strtprofile)$ are strictly positive due
to our assumption.  For ease of reasoning, we assume without loss of
generality that $s^+(\strtprofile) \geq s^-(\strtprofile) > 0$. The
other case is symmetric.

\Cref{alg:proandcon} iterates through the agents and if they expressed positive support for $\omega$, then it lowers the number of votes they cast as much as possible by \say{consuming} from the negative support. Since $s^+(\strtprofile) \geq s^-(\strtprofile)$, all negative votes are reduced to $0$. 
\begin{algorithm}[h!]
	\begin{algorithmic}
		\State $D \gets s^-(\strtprofile)$
		\For{$i \in \mathscr{C}$}
		\If{$\allvotes{i} > 0$} 
		\State $x_i \gets max(0, \allvotes{i}-D)$
		\State $D \gets D - (v_i - x_i)$
		\Else 
		\State $x_i \gets 0$
		\EndIf \EndFor
	\end{algorithmic}
	\caption{Opposing votes cancellation}
	\label{alg:proandcon}
\end{algorithm}
\begin{example}
	Let $\strtprofile = (7, -9, 5, -1, 1)$. Then\cref{tbl:proandconex} shows the initial value of $D$ is $s^-(\strtprofile) = 10$ and the iterations.
	
	\begin{table}[t]
		\centering
		\caption{Example of the first beneficial collusion strategy (\cref{alg:proandcon}). Votes shown only for the relevant outcome. Each row contains values after one complete iteration. }
		\begin{tabular}{c c c c c c c }
			\toprule
			Agents & $1$ & $2$ & $3$ & $4$ & $5$ & \\
			\midrule
			$\strtprofile$ & 7 & -9 & 5 & -1 & 1 & $D = 10$ \\
			$x$ & 0 & -9 & 5 & -1 & 1 & $D = 3$ \\
			$x$ & 0 & 0 & 5 & -1 & 1 & $D = 3$ \\
			$x$ & 0 & 0 & 2 & -1 & 1 &$D=0$ \\
			$x$ & 0 & 0 & 2 & 0 & 1 & $D=0$\\
			\bottomrule
		\end{tabular}
		\label{tbl:proandconex}
	\end{table}
\end{example}

\begin{restatable}{theorem}{collusiononed}
	\label{th:collusion_one_direction} When an outcome has both pro and contra votes in $\strtprofile$, \cref{alg:proandcon} generates a strictly-beneficial collusion strategy which satisfies \cref{eq:totalconstant}, \eqref{eq:indiv_rational} and \eqref{eq:strict}.
\end{restatable}

Proof in \cref{appendix:collusion_one_direction}.

\noindent
\paragraph{Method 2: When all votes are in the same direction (per outcome)}
The following procedure can be applied even if there are opposing votes for a subset of the outcomes, but to avoid confusion between the number of votes $\votes{i}{\omega}$ and the absolute intensity of support $|\votes{i}{\omega}|$, we will assume that all values are positive.

\begin{table}[t]
	\centering      
	\caption{The simplified strategy of Method 2, before colluding (left) and after colluding (right).}
	\label{fig:simple_strat}
	\begin{subtable}[t]{0.2\textwidth}
		\centering
		\begin{tabular}{ c c c }
			\toprule
			Agents & \multicolumn{2}{c}{$\strtprofile$} \\
			\cmidrule{2-3}
			& $\omega$ & $\phi$\\
			\midrule
			$A$ & $a$ & $b$ \\
			$B$ & $c$ & $d$ \\
			\bottomrule 
		\end{tabular} 
	\end{subtable}
	\begin{subtable}[t]{0.2\textwidth}
		\centering
		\begin{tabular}{c c c}
			\toprule
			Agents & \multicolumn{2}{c}{$x$} \\
			\cmidrule{2-3}
			& $\omega$ & $\phi$ \\
			\midrule
			$A$ & $a-1$\tikzmark{a} & $b+1$\tikzmark{b} \\
			$B$ & $c+1$\tikzmark{c} & $d-1$\tikzmark{d} \\
			\bottomrule
		\end{tabular}
		\begin{tikzpicture}[overlay, remember picture, shorten >=.5pt, shorten <=.5pt, transform canvas={yshift=.25\baselineskip}]
	\draw [->] ({pic cs:a}) [bend left] to ({pic cs:c});    
	\draw [->] ({pic cs:d}) [bend right] to ({pic cs:b});
\end{tikzpicture}
	\end{subtable}
\end{table}

We illustrate this method using a $2\times 2$ example, with two colluding agents and two outcomes. The set of agents in the coalition is $\mathscr{C} = \{A,B\}$ and their voting strategies $\allvotes{i}$ are presented in \cref{fig:simple_strat}, with $\strtprofile$ satisfying $a > b$ and $c < d - 1$. If $A$ and $B$ exchange one vote on $\omega$ for one vote on $\phi$ as shown, then the total number of votes per outcome stays constant and their utilities change according to \cref{eq:utils_in_2x2}; thus, no agent is disadvantaged by participating and at least one benefits:

\begin{equation}
	\label{eq:utils_in_2x2}
	\begin{split}
		pay^A(x)=(a-1)^2+(b+1)^2 &\leq a^2 + b^2=pay^A(\strtprofile), \\
		pay^B(x)=(c+1)^2+(d-1)^2 &< c^2+d^2=pay^B(\strtprofile).
	\end{split}
\end{equation}

\cref{th:cycle} says that this \say{exchange one vote for $\omega$ for one vote for $\phi$} design can be applied to a set of agents of arbitrary size provided that the individual strategies before colluding satisfy the following condition:

We define $G(\strtprofile) = \langle \Omega, E \rangle$ to be the \textit{claimed preference graph} of the strategy profile $\strtprofile$. The graph contains a directed edge from outcome $\omega$ to outcome $\phi$ if there is an agent that casts more votes for $\omega$ than $\phi$. We say an edge $(\omega \rightarrow^i \phi)$ is \textit{strictly beneficial} for $i$ if and only if $\votes{i}{\omega} < \votes{i}{\phi} - 1$. Formally, the edge set $E$ is defined by:
$$E = \{\omega \rightarrow^i \phi \mid \votes{i}{\omega} < \votes{i}{\phi} \}.$$

\begin{restatable}{theorem}{cycle}
	\label{th:cycle}
	If all votes are in the same direction (per outcomes) and $G(\strtprofile)$ contains a cycle with a strictly beneficial edge, then there is a set of agents that can benefit from strictly-beneficial collusion.
\end{restatable}

\begin{proof}[Proof sketch]
	The detailed proof is in \cref{appendix:cycle}. If the graph has a cycle, it also has a simple cycle. The coalition is constructed by selecting agents that are part of the fixed simple cycle in $G(\strtprofile)$ and shows that if vote \say{exchanges} are according to the edges of this cycle, the coalition satisfies all the required conditions.
\end{proof}

Observe that the condition in \cref{th:cycle} is sufficient, but not necessary. It may be the case that the graph is acyclic but the agents can still benefit by colluding. For example, in \cref{fig:nocycle_coal}, the prices incurred by every agent decrease after participating in
the coalition:
\begin{equation*}  
	\begin{split}
		\sum_{\omega} &(\votes{A}{\omega})^2 {=}6^2{+}5^2{+}4^2{=}77 
		 {>} \sum_{\omega} (\xvotes{A}{\omega})^2{=}7^2{+}4^2{+}3^2{=}74. 
		\\ \sum_{\omega} & (\votes{B}{\omega})^2{=}10^2{+}1^2{+}1^2{=}102
		 {>}\sum_{\omega} (\xvotes{B}{\omega})^2{=}9^2{+}2^2{+}2^2{=}89.
	\end{split}
	\label{eq:nocycle_pay}
\end{equation*}
\begin{table}[t]
	\centering
	\caption{Coalition on a non-cyclical preference graph, before colluding (left) and after colluding (right).}
	\label{fig:nocycle_coal}
	\begin{subtable}[b]{0.2\textwidth}
		\centering
		\begin{tabular}{ c c c c  }
			\toprule
			Agents & \multicolumn{3}{c}{$\strtprofile$}\\
			\midrule
			$A$ & $6$ & $5$ & $-4$ \\
			$B$ & $10$ & $1$ & $-1$ \\
			\bottomrule 
		\end{tabular}
	\end{subtable}
	\quad
	\begin{subtable}[b]{0.2\textwidth}
		\centering
		\begin{tabular}{c c c c }
			\toprule
			Agents & \multicolumn{3}{c}{$x$}\\
			\midrule
			$A$ & $7$ & $4$ & $-3$ \\
			$B$ & $9$ & $2$ & $-2$ \\
			\bottomrule 
		\end{tabular}
	\end{subtable}
\end{table}

In practical settings, almost all graphs will contain a cycle that involves a
strictly beneficial edge because the preferences of agents will not align perfectly. This suggests that the strategy presented in this section is widely applicable.

\vspace*{1ex}\noindent\textbf{Discussion:}
Collusion is possible in most voting mechanisms, to different degrees. Regarding QV, it can be an issue for elections with a small number of agents, if an agent with a strong preference convinces (by reason or even bribing) other agents to split their votes with it, given that their marginal vote price is smaller. In this case, every agent has a large impact on the outcome of the election. The main way to avoid such manipulations is to ensure complete voting secrecy, in the sense that no one can prove what they voted. The most popular solution is E2E voting, which is rigorously discussed and applied to QV in \cite{secureqv}.

On the other hand, \citet{robustness} proves that despite QV's susceptibility to collusion, in order for the coalition to spoil QV's utilitarian efficiency, the coalition size has to be significantly large compared to the society's size. Due to coordination problems, this is unlikely for large-scale elections (e.g., $\sqrt{|N|}$, if $N$ represents California), and even more so in our stricter definition of collusion.


\newcommand{\titleofCom}[0]{Computational Complexity}
\section{\titleofCom}
\label{sec:complexity}
In this section, we show that determining whether a strategy profile is a (pure) NE is polynomial-time solvable for both no-budget and fixed-budget multiple-issue QV.


\subsection{No budget limit}
\label{sec:complexity_no_budget}

Given the instance $\langle N, \Omega, u, v \rangle$, with utilities $u \in \Z^{|N| \times |\Omega|}$ and strategy profile $\strtprofile \in \Sigma^{|N|}$, we aim to verify if $\strtprofile$ forms a pure NE. In practice, no agent has an infinite budget, so we can assume the maximal number of votes placed on any outcome is bounded by a (large) constant. Still, the size of the strategy set $\Sigma$ is exponential in $|\Omega|$. 

\begin{definition}
	Define $\QVNE$ to be the language containing $\langle N, \Omega, u, \strtprofile \rangle$-instances, each an election with society $N$, outcomes $\Omega$, preferences described by the utility matrix $u$, and where $\strtprofile$ is a NE:
	$$ \QVNE = \{\langle N, \Omega, u, \strtprofile \rangle \mid \strtprofile \text{ represents a (pure) NE} \}.$$
	
\end{definition}

The naive way to verify whether $\strtprofile$ is a pure NE involves iterating through the entire set of strategies $\Sigma$ for each agent and checking if the alternative strategy provides greater utility (i.e., if the agent has a beneficial deviation). Alternatively, this problem can be solved by a non-deterministic Turing machine by guessing the certificate (i.e., the agent and its alternative strategy of length $|\Omega|$) and verifying in polynomial time if the resulting utility is greater than the original. This indicates $\QVNE \in$ NP. However, we are in fact able to present a polynomial time algorithm to solve this problem.

First, we define the decision problem $\DEVIATE$, and show that it is solvable in polynomial time in \cref{th:deviate}.

\begin{definition} $\DEVIATE$, contains instances of the shape $\langle N, \Omega, u, \strtprofile, i \rangle$, such that, in the election with society $N$, outcomes $\Omega$ and utility matrix $u$, agent $i$ can beneficially deviate from the strategy profile $\strtprofile$:
  \begin{equation*}
  \begin{split}
      \DEVIATE &=\\  \{\langle N, \Omega, &u, \strtprofile, i \rangle \mid \text{agent $i$ has a beneficial deviation}  \}
  \end{split}
  \end{equation*}
\end{definition}

\begin{restatable}{theorem}{deviate}
	$\DEVIATE \in \mathrm{P}$.
	\label{th:deviate}
\end{restatable}

\begin{proof}[Proof sketch]
	Instead of iterating through all of agent $i$'s possible strategies, we construct a subset of strategies $\PBR$, with size polynomial in $|N|$ and $|\Omega|$, and show that $\PBR$ is guaranteed to contain a best response. Assume agent $i$'s best response is to bring all outcomes in the set $W$ to $V + 1$ votes. Then, it will have to cast positive votes for some outcomes $W^+$ and negative votes for some outcomes $W^-$. Consider the sizes of these sets to be fixed to $\sizewneg = |\Wneg|$ and $\sizewpos = |\Wpos|$. We prove that the optimal $\Wneg$ is independent from $\Wpos$ or $V$ and that it can be found in $\O(\sizewneg \log \sizewneg)$.
	
	Each member of $\Wpos$ influences the value of $U^i$ with a quadratic function that also depends on $V$. Note that if $V$ is in a specific interval, it characterises which are the best $\sizewpos$ outcomes that should be included in $\Wpos$. There are $\O(|\Omega|^2)$ intervals described by the quadratic terms and, hence, that many options for $\Wpos$.
	Now, we can recreate $\woptim$ from $\Wpos$ and $\Wneg$ and obtain the optimal $V$ for this set using optimization techniques. In total, there is a small number of variables taking polynomially many values. We present a rigorous proof and algorithm in \cref{appendix:deviate}.
\end{proof}
\begin{restatable}[Completeness]{lemma}{completeness}
	\label{lemma:completeness} 
	If agent $i$ has a beneficial deviation, then there exists $\sigma^* \in \PBR$ such that $U^i(\allvotes{-i}, \sigma^*) > U^i(\allvotes{-i}, \allvotes{i})$ (i.e., $\sigma^*$ is a best response).
\end{restatable}
\begin{proof}[Proof sketch]
	$\sigma^*$ corresponds to a set of variables $S$ which the algorithm iterates over. Therefore, since $\PBR$ contains the best responses for each combination of such variables, the strategy selected for $S$ will be the best response. This is also formally shown in \cref{appendix:completeness}.
\end{proof}
\begin{lemma}[Soundness]
	\label{lemma:soundness}
	If agent $i$ does not have a beneficial deviation, then $\forall \sigma \in \PBR$,
	$U^i(\allvotes{-i}, \sigma) \leq U^i(\allvotes{-i}, \allvotes{i})$.
\end{lemma}

\begin{proof}
	If agent $i$ does not have a beneficial deviation,  then \\$U^i(\allvotes{-i}, \sigma) \leq U^i(\allvotes{-i}, \allvotes{i})$, $\forall \sigma \in \Sigma$. Since, $\PBR \subseteq \Sigma$, the result is immediate.
\end{proof}

Since $\DEVIATE$ is polynomial-time solvable, we can prove that $\QVNE$ can also be decided in polynomial time.

\begin{restatable}{theorem}{qvne}
	\label{th:qvne}
	$\QVNE \in \mathrm{P}$.
\end{restatable}
\begin{proof}
	Let $I$ be the input instance to $\QVNE$. Run the decider for $\DEVIATE$ with input $\langle I, i \rangle$ for each agent $i \in N$. If at least one agent has a beneficial deviation, then $I \notin \QVNE$. Hence, $\QVNE$ is Turing reducible to $\DEVIATE$ using a polynomial number of calls. Using \cref{th:deviate}, we therefore see $\QVNE \in$ P.
\end{proof}

\subsection{Fixed budget}
\label{sec:complex_fixed_budget}
We now consider the same question as in the previous subsection, but for the fixed-budget variation of QV. In this setting, all agents receive a budget of $B$ credits to spend on voting. We can assume $B$ is at most polynomial in the size of the population, as, in practice, any budget which is much larger than the number of issues would encourage agents to extreme vote allocations.
Recall (in \cref{def:fixedbudgetQV}) that agent $i$'s total utility is :
\begin{gather*}
	U^i(\strtprofile) = \frac{1}{|W|}\sum_{\omega \in W} \util{i}{\omega}, \quad \text{ where } \quad \sum_{\omega \in \Omega} (\votes{i}{\omega})^2 \leq B.
\end{gather*}

We now introduce $\DEVIATE_B$ and $\QVNE_B$:

\begin{definition} Define $\DEVIATE_B$ to be the set of instances representing elections with the set of agents $N$, outcomes $\Omega$, and utility matrix $u$, where the agent $i$ has a beneficial deviation from $\strtprofile$, given that all strategies are limited by the number of credits $B$. 
  \begin{equation*}
  \begin{split}
		\DEVIATE_B & = \\ \{\langle B, N, &\Omega, u, \strtprofile, i \rangle \mid \text{agent $i$ has beneficial deviation}\}. 
  \end{split}
\end{equation*}
\end{definition}

\begin{definition}
	Let $\QVNE_B$ be the problem of determining if a given strategy profile $\strtprofile$ forms a NE in an election with society $N$, budget $B$, and agent preferences described by $u$ on outcomes $\Omega$. Formally,
	\begin{equation*}
		\QVNE_B = \{\langle B, N, \Omega, u, \strtprofile \rangle \mid \strtprofile \text{ is a pure-strategy NE} \}.
	\end{equation*}
\end{definition}

\Cref{th:dev_b} shows a tighter bound on the complexity of $\DEVIATE_B$, than the naive \say{guess and check} algorithm.

\begin{restatable}{theorem}{devb}
	$\DEVIATE_B \in \mathrm{P}$.
	\label{th:dev_b}
\end{restatable}

\begin{proof}[Proof sketch]
	Agent $i$ can not modify the number of votes for the winning outcome with more than $\O(\sqrt{B})$ votes, due to the budget constraint, and we can therefore iterate through the possible values of $V$.
	Define $dp[n][p][b]$ to be the maximum utility sum agent $i$ can obtain by partitioning the first $n$ outcomes into 
	$p$ possible winning and $n-p$ losing, all by using a budget of at most $b$. The recurrence can be computed in $\O(1)$ and the final algorithm runs in $\O(|\Omega|^2 B \sqrt{B})$. The complete proof is in \cref{appendix:dev_b}.
\end{proof}

Finally, we show that $\QVNE_B$ is  polynomial-time decidable, just like its no-budget counterpart, $\QVNE$.

\begin{theorem}
	$\QVNE_B \in \mathrm{P}$.
\end{theorem}

\begin{proof}
	Analogous to \cref{th:qvne}, using that $\DEVIATE_B \in$ P (\cref{th:dev_b}).
\end{proof}

\section{Conclusions and Future Work}

Filling a gap in the topic of social decision-making, our contribution addresses different aspects of QV, from game-theoretic to algorithmic. We present the first mathematical model for fixed-budget multiple-issue QV, inspired by the no-budget multiple-issue~\cite{multiplealternative} and binary fixed-budget \cite{quad_election_law} variations. Then, we examine both fixed-budget and no-budget QV against established voting mechanisms and criteria, discuss advantages and defects of QV and show two possible collusion strategies. Afterwards, we adopted a computational perspective and showed that determining if a given strategy represents a pure NE is polynomial-time decidable for the two multiple-issue QV alternatives. We hope that our paper provides necessary theoretical groundwork and encourages future research of QV in practical settings. 

One such promising application is resource allocation \cite{res_alloc_1, res_alloc_2, res_alloc_3, Gautier_Lacerda_Hawes_Wooldridge_2023}. Multi-agent systems often have \textit{constraints}, such as the daily energy available to a team of Mars rovers or the total time available on the CPU. Using fixed-budget multiple-issue QV for this setting would provide an alternative to classic decision-making mechanism, one which can leverage the theoretical properties of QV we presented.

\bibliography{bibliography}







\appendix

\onecolumn
\begin{center}
\textbf{\huge Supplementary Material: Fixed-budget and Multiple-issue Quadratic Voting}
\end{center}

\section{\titleofProp}
\label{appendix:properties}

\properties*
\begin{proof} Each criterion is analysed separately:
\begin{itemize}
    \item \textbf{Majority safe}: Consider $|N| = 3$ agents and outcomes $\Omega = \{ \omega, \phi \}$. If two agents place 1 vote each for $\omega$ and the other agent places $3$ votes on $\phi$, then $\phi$ wins even if the majority prefers $\omega$.
    \item \textbf{Consistency}: If an outcome $\omega$ has the maximal number of votes in all elections part of $\mathscr{E}$ (i.e $s^e_{\omega} \geq s^e_{\psi}, \forall e \in \mathscr{E}, \psi \in \Omega$), then the sum of votes in the union election satisfies $\sum_{e} s^e_{\omega} \geq \sum_{e} s^e_{\psi}, \forall \psi \in \Omega$. Therefore, $\omega$ will win the election formed by summing all the elections in $\mathscr{E}$.
    
    \medskip
      When discussing IIA and clones for rated methods, it is assumed
      that agents evaluate each outcome independently. Otherwise, if
      the outcomes are compared to one another, rated methods become
      more like ranked systems.
    
    \item \textbf{Clone independence}: Consider the cases where the set of possible winning
      outcomes has size 1, as different tie-breaking criteria may be
      applied to all other voting systems. Following the assumption made for approval voting
      \cite{cloneproof} and range voting \cite{range_ICC_AFB}, strategic manipulation is forbidden and hence the votes for other candidates are not influenced by the introduction of the clone. Assume a non-winning clone $\gamma$ of $\omega$ is
      introduced into the ballot. Since $\gamma$ is a clone, it holds that
      $\util{i}{\gamma} = \util{i}{\omega} - \epsilon$, for all agents
      $i \in N$ and a very small constant
      $\epsilon > 0$. Then, when agents vote truthfully,
      $\votes{i}{\gamma}=\votes{i}{\omega}$ and the ordering of the
      outcomes per agent is the same as in the election without
      $c$. It follows that the overall ordering by $\totalvotes{\phi}$
      and the winner does not change ($\gamma$ is a non-winning clone).

    \item \textbf{IIA}: If $\omega$ is preferred over $\phi$, then
      $\totalvotes{\omega} \geq \totalvotes{\phi}$. When $\psi$ is introduced and
      the votes cast for $\omega$ and $\phi$ remain constant, then
      $\totalvotes{\omega}$ and $\totalvotes{\psi}$ are unaffected, due to the
      separate evaluation assumption, and, therefore, $\omega$ has at least
      as many chances of winning as $\phi$.

    \item \textbf{NFB}: Assume the schema defines $\alpha=1$, the set of agents $N = \{1,2\}$, the
      set of outcomes $\Omega=\{\omega,\phi,\psi\}$. If the utilities of agent $1$
      and votes cast by $2$ are as displayed in \cref{tbl:fb_ex}, then
      agent $1$ has no incentive to put enough votes to make his
      favorite, $\psi$ win. To close the gap between $\omega$ and $\psi$, it has
      to pay at least $2(\frac{100+10}{2})^2$, which makes its final
      utility negative. If, instead, $\allvotes{1}=(0,4,0)$, agent $1$
      obtains $U^1(\strtprofile)=900-16=884$. Therefore, $1$ has an
      incentive to rate another candidate ($\phi$) higher than his
      favorite ($\psi$).

    \begin{table}[h]
    \centering
    \caption{Favourite betrayal example. The outcome preferred by agent $1$ has many negative votes from agent $2$, and since $1$ does not afford to cast sufficient votes to support $\psi$, it betrays $\psi$ and votes for $\phi$ instead. }
    \begin{tabular}{c c c c}
        \toprule
         & $\omega$ & $\phi$ & $\psi$ \\
         \midrule
        $u^1$ & 0 & 900 & 910 \\
        $\allvotes{2}$ & 10 & 7 & -100 \\
        $\allvotes{1}$ & 0 & 4 & 0 \\
        \bottomrule
    \end{tabular}
    \label{tbl:fb_ex}
    \end{table}

    For the fixed-budget QV, the same argument can be made by simply limiting the budget to $B = 16$. Agent $1$ does not have enough credits to place on $\psi$, but it can make $\phi$ win.
\end{itemize}
\end{proof}

\subsection{Section \ref{sec:collusion}}

\collusiononed*
\begin{proof}
\label{appendix:collusion_one_direction}

Let $x$ be the output of \cref{alg:proandcon} for the input strategy profile $\strtprofile$. For \cref{eq:totalconstant}: $s^+(x) = s^+(\strtprofile) - D$ and $s^-(x) = 0$, so $s^+(x) - s^-(x) = s^+(\strtprofile) - s^-(\strtprofile)$. This implies $\sum_{i \in \mathscr{C}} x_i = \sum_{i \in \mathscr{C}} \allvotes{i}$. Since the votes for all other outcomes remain unchanged, the result is immediate.

For \cref{eq:indiv_rational} and \eqref{eq:strict}:
The invariant of the for loop is $D \geq 0$.
Therefore, $|x_i| \leq |\allvotes{i}|, \forall i \in \mathscr{C}$ and the inequality is strict for the first positive value since $D_{init} = s^-(\strtprofile) > 0$. Hence, $x_i^2 \leq (\allvotes{i})^2$ for all agents and at least one is better off as part of the coalition.
\end{proof}

\cycle*
\begin{proof}
\label{appendix:cycle}
 If $G(\strtprofile)$ contains one such cycle, it also contains a simple cycle $\mathscr{C}_{\Omega}$ containing a strictly beneficial edge. For each edge $(\omega \rightarrow^i \phi) \in \mathscr{C}_{\Omega}$, agent $i$ increases the number of votes for $\omega$, $\xvotes{i}{\omega}=\votes{i}{\omega}+1$, and decreases the support for $\phi$, $\xvotes{i}{\phi}=\votes{i}{\phi} - 1$. The coalition $\mathscr{C}$ is the set of agents annotated on the edges of the cycle, that is $\mathscr{C}=\{i \mid \exists \omega, \phi : (\omega \rightarrow^i \phi) \in \mathscr{C}_{\Omega}\}$. 
 
 Consider, for instance, \cref{fig:coll_before} for which the induced graph contains the simple cycle $\mathscr{C}_{\Omega} = \{\psi \rightarrow^A \omega, \omega \rightarrow^B \phi, \phi \rightarrow^C \psi\}$. Then, all agents are involved in the coalition and, by applying the strategy above for this cycle, their new strategies $x$ are depicted in \cref{fig:coll_after}.
 
 \begin{table}[t]
 \centering
      \caption{Example of the generalised strategy. The black arrows correspond to the vote transfers, while the red arrows correspond to the edges in the cycle $\mathscr{C}_{\Omega}$.}
 \label{fig:gen_cycle}
 \begin{subtable}[b]{0.3\textwidth}
 \centering
 \begin{tabular}{c c c c}
        \toprule
        Agents & \multicolumn{3}{c}{$\strtprofile$} \\
        \cmidrule(l){2-4}
         & \multicolumn{1}{c}{$\omega$} & \multicolumn{1}{c}{$\phi$} & \multicolumn{1}{c}{$\psi$}\\
        \midrule
        $A$ & $1$ & $0$ & $3$ \\
        $B$ & $3$ & $1$ & $0$ \\
        $C$ & $0$ & $6$ & $2$ \\
        \bottomrule
      \end{tabular}
      \caption{Before colluding}
      \label{fig:coll_before}
    \end{subtable}
      \quad
    \begin{subtable}[b]{0.3\textwidth}
    \centering
        \begin{tabular}{c c c c}
        \toprule
         Agents & \multicolumn{3}{c}{$x$} \\
            \cmidrule{2-4}
         & $\omega$ & $\phi$ & $\psi$\\
        \midrule
        $A$ & \tikzmarknode{a2}{$2$} & $0$ & \tikzmarknode{c1}{$2$} \\
        $B$ & \tikzmarknode{a1}{$2$} & \tikzmarknode{b2}{$2$} & $0$ \\
        $C$ & $0$ & \tikzmarknode{b1}{$5$} & \tikzmarknode{c2}{$3$} \\
        \bottomrule
      \end{tabular}  
      \begin{tikzpicture}[overlay, remember picture, shorten >=.5pt, shorten <=.5pt]
        \draw [->] (a1.west) [bend left] to (a2.west);    
        \draw [->] (b1.east) [bend right] to (b2.east); 
        \draw [->] (c1.east) [bend left] to (c2.east);

        \draw [color=red,->] (a2.north east) [bend left] to (c1.north west);
        \draw [color=red,->] (b2.south west) [bend left] to (a1.south east);
        \draw [color=red,->] (c2.south west) [bend left] to (b1.south east);
    \end{tikzpicture}
    
      \caption{After colluding}
      \label{fig:coll_after}
      \end{subtable}
 \end{table}

We will now show that our construction meets the three requirements: unchanged ballot-count, individual rationality, and strict benefit for one agent, namely \cref{eq:totalconstant}, \eqref{eq:indiv_rational} and \eqref{eq:strict}. Consider agent $i \in \mathscr{C}$, $\omega$ and $\phi$ such that $(\omega \rightarrow^i \phi) \in \mathscr{C}_{\Omega}$. From the definition of the graph, $\votes{i}{\omega} < \votes{i}{\phi}$. Because $\mathscr{C}_{\Omega}$ is a simple cycle, exactly $2$ values differ between $\allvotes{i}$ and $x^i$, at indices $\omega$ and $\phi$. Then, \cref{eq:cycle_indiv_rat} proves that individual rationality is satisfied for all agents part of the coalition.

\begin{equation}
\label{eq:cycle_indiv_rat}
    \begin{split}
       \sum_{\psi \in \Omega} (\xvotes{i}{\psi})^2 &= (\xvotes{i}{\omega})^2 + (\xvotes{i}{\phi})^2 + \sum_{\psi \in \Omega \setminus \{\omega, \phi\}} (\xvotes{i}{\psi})^2 \\
         &=  (\votes{i}{\omega}+1)^2 + (\votes{i}{\phi}-1)^2 + \sum_{\psi \in \Omega \setminus \{\omega, \phi\}} (\votes{i}{\psi})^2 \\
         &= 2(\votes{i}{\omega}-\votes{i}{\phi}+1)+ \sum_{\psi \in \Omega} (\votes{i}{\psi})^2 \\
         &\leq \sum_{\psi \in \Omega} (\votes{i}{\psi})^2.
    \end{split}
\end{equation}
The inequality is strict for the strictly beneficial edge, hence meeting \cref{eq:strict}.

The total number of votes per outcome $\totalvotes{\omega}$ is not affected for $\omega \notin \mathscr{C}_{\Omega}$. Consider outcome $\omega \in \mathscr{C}_{\Omega}$ involved in the vote transfer; as it is part of a directed simple cycle, it has exactly one incoming edge $\rightarrow^{i_\omega}$ and one outgoing edge $\rightarrow^{j_\omega}$, so $\totalvotes{\omega}$ remains unchanged after applying the collusion strategy:
\begin{equation}
\begin{split}
\totalvotes{\omega} = \sum_{i \in N} \votes{i}{\omega} &= \votes{i_\omega}{\omega} + \votes{j_\omega}{\omega} + \sum_{i \in N \setminus \{i_\omega,j_\omega\}} \votes{i}{\omega} \\
            &= (\xvotes{i_\omega}{\omega}+1) + (\xvotes{j_\omega}{\omega}-1) + \sum_{i \in N \setminus \{i_\omega,j_\omega\}} \xvotes{i}{\omega} %
            = \sum_{i \in N} \xvotes{i}{\omega}. \\
\end{split}
\end{equation}

Thus, all three requirements for a successful coalition are fulfilled.
\end{proof}

\section{\titleofCom}

\subsection{Section \ref{sec:complexity_no_budget}}
\deviate*
\begin{proof}
\label{appendix:deviate}
We show the result by constructing a set of possible strategies $\PBR$, with size polynomial in $|N|$ and $|\Omega|$, that is guaranteed to contain a best response strategy. For simplicity, assume that the outcomes are numbered from $1$ to $|\Omega|$. Define $W=\argmax_\omega \totalvotes{\omega}$ to be the set of possible winners after $i$'s ballot is included. Let $\sumvotes{i}{\omega} = \sum_{j \neq i} \votes{j}{\omega}$ be the total number of votes option $\omega$ has from all other agents. 

Assume agent $i$ wants to influence the votes such that each possible winner has $V + 1$ votes in total and decrease the votes of the other outcomes with $\sumvotes{i}{\omega} \geq V + 1$ to $V$. Notice that agent $i$ might want to decrease the number of votes for some outcomes it wants in $W$, in order to pay less overall. Take for example, a situation in which $\Omega$ has $2$ outcomes, with $10$ and $4$ votes from the other agents respectively, and the agent wants both outcomes to win. In this case, it should bring both to $7$ votes in order to pay as little as possible. Then, if $v^{*,i}$ is $i$'s goal strategy, it satisfies:
$$
\totalvotes{\omega} = \sumvotes{i}{\omega} + \optimvotes{i}{\omega} =
\begin{cases}
V + 1, & \text{if $\omega \in W$} \\
\max(V, \sumvotes{i}{\omega}), & \text{otherwise.}
\end{cases}
$$

Hence, the strategy of agent $i$ is determined by $V$ and $W$,
$$
\optimvotes{i}{\omega}=
\begin{cases}
V + 1 - \sumvotes{i}{\omega}, & \text{if $\omega \in W$} \\
V - \sumvotes{i}{\omega}, & \text{if $\omega \not\in W$ and $\sumvotes{i}{\omega} > V$ } \\
0, & \text{otherwise.} \\
\end{cases}
$$

Partition $W$ into two sets $\Wneg = \{\omega \in W \mid \optimvotes{i}{\omega} < 0\}$ and $W^+ = \{\omega \in W \mid \optimvotes{i}{\omega} \geq 0\}$. Assume without loss of generality that the outcomes are ordered such that $s^{-i}$ is in descending order and let $m$ be the number of outcomes with strictly more votes than $V$ (excluding the votes from agent $i$). Therefore, $\sumvotes{i}{m} > V \geq \sumvotes{i}{m+1}$, where we can consider $\sumvotes{i}{0} = \infty$ and $\sumvotes{i}{|\Omega|+1} = -\infty$, and $\{1,\dots,m\} \cap W = \Wneg$. \Cref{fig:example_strategy} illustrates an example strategy.


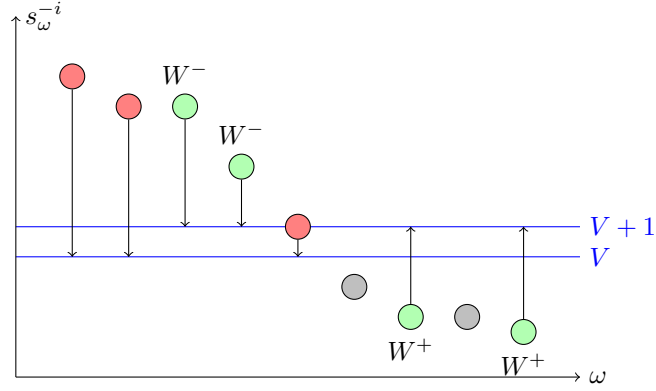
\begin{figure}[t]
\centering
\def\lim{3}
\begin{tikzpicture}[yscale=0.4, xscale=0.75]
  \tikzstyle{nodest}=[circle,draw,fill=red!50]
  \draw[-] [color=blue] (0,\lim) -- (10,\lim) node[right] {$V$};
  \draw[-] [color=blue] (0,\lim+1) -- (10,\lim+1) node[right] {$V+1$};
  \draw[->] (0,-1) -- (10,-1) node[right] {$\omega$};
  \draw[->] (0,-1) -- (0,11) node[right] {$\sumvotes{i}{\omega}$};

  \node[style=nodest] (1) at (1,9) {};
  \node[style=nodest] (2) at (2,8)  {};
  \node[style=nodest, fill=green!30] (3) at (3,8) [label=above:$\Wneg$]  {};
  \node[style=nodest, fill=green!30] (4) at (4,6) [label=above:$\Wneg$] {};
  \node[style=nodest] (5) at (5,4)  {};
  \node[style=nodest, fill=gray!50] (6) at (6,2)  {};
  \node[style=nodest, fill=green!30] (7) at (7,1) [label=below:$\Wpos$] {};
  \node[style=nodest, fill=gray!50] (8) at (8,1)  {};
  \node[style=nodest, fill=green!30] (9) at (9,0.5) [label=below:$\Wpos$] {};

  \foreach \x in {1,2,5} \draw[->] (\x) -- (\x,\lim); 
  \foreach \x in {3,4,7,9} \draw[->] (\x) -- (\x,\lim+1); 
\end{tikzpicture}
\caption{Strategy based on $V$ and $W = \Wneg \cup \Wpos$. The green nodes represent outcomes that become the possible winners after agent $i$ casts their votes. The red nodes are the outcomes that receive negative support from agent $i$, in order to be excluded from winning the election. The agent does not cast any votes for the gray outcomes.}

\label{fig:example_strategy}
\end{figure}

Substituting in \cref{eq:indiv_util}, where we ignored the refund term as mentioned:
\begin{equation}
\label{eq:W_util}
    U^i(V, W, u) = \frac{1}{|W|}\sum_{\omega \in W} \util{i}{\omega} - \alpha\bigl[\sum_{\omega \in W}(\sumvotes{i}{\omega} - (V + 1))^2 + \sum_{\omega \in \{1,\dots,m\}\setminus W} (\sumvotes{i}{\omega} - V)^2 \bigr]
\end{equation}

\newcommand{\voptim}[0]{V^*(W)}
\newcommand{\vmoptim}[0]{V^*_m(W)}
For a more concise derivation, we temporary extend the domain of $V$ to $\R$ instead of $\Z$\footnote{We will restrict it back towards the end of the proof. \label{V_domain}}. Let $\vmoptim$ be a value of $V$ that minimises the payment incurred by agent $i$ such that the final set of winners is $W$. To simplify the notation, we will just use $\voptim$, but it is a function of $m$. Now, let us find the $\voptim$ that maximises $i$'s utility.

\begin{align*}
    \frac{\partial U^i}{\partial V} &= -\alpha \bigl( 2|W|V - 2\bigl(\sum_{\omega \in W} (\sumvotes{i}{\omega} - 1)\bigr) + 2V(m-|\Wneg|) - 2\sum_{\omega \in \{1,\dots, m\}\setminus W} \sumvotes{i}{\omega} \bigr) \\
    &= -2\alpha \Bigl( (m+|\Wpos|)V - \bigl(\sum_{\omega=1}^{m} \sumvotes{i}{\omega} + \sum_{\omega \in \Wpos} \sumvotes{i}{\omega}\bigr) - |W| \Bigr) 
\end{align*}

It holds that $\frac{\partial U^i}{\partial V}(\voptim) = 0$ so

\begin{equation}
\begin{split}
\vmoptim &= \frac{\sum_{\omega \in W} (\sumvotes{i}{\omega} - 1) + \sum_{\omega \in \{1,\dots, m\}\setminus W} \sumvotes{i}{\omega}}{m + |\Wpos|}  \\
 &= \frac{(\sum_{\omega = 1}^m \sumvotes{i}{\omega}) + (\sum_{\omega \in \Wpos}\sumvotes{i}{\omega}) - |W|}{m + |\Wpos|}
\label{eq:v_optim}
\end{split}
\end{equation}
The equation above shows how to obtain the best $V$ from an arbitrary $W$. One can notice that $\voptim$ it not symmetric in $\Wneg$ and $\Wpos$. The contents of $\Wneg$ do not influence $\voptim$, whereas $\Wpos$ plays a crucial role in determining the value.

\newcommand{\wnegoptim}[0]{W^{*,-}(\sizewneg, \sizewpos)}
\newcommand{\wposoptim}[0]{W^{*,+}(\sizewneg, \sizewpos)}
\newcommand{\wjposoptim}[1]{W^{*,+}_{#1}(\sizewneg, \sizewpos)}
\newcommand{\wmoptim}[0]{W^*_m(\sizewneg, \sizewpos)}
\newcommand{\wmnegoptim}[0]{W^{*,-}_m(\sizewneg, \sizewpos)}
\newcommand{\wjmposoptim}[1]{W^{*,+}_{m,{#1}}(\sizewneg, \sizewpos)}

We now turn our attention to finding the optimal set of winners $W$. Let $\woptim$ be the set of final winners $W$ with size $\sizewpos + \sizewneg$ which yields the maximal $U^i$ when agent $i$ chooses strategy $a^{*,i}(\voptim, W, u)$, where $|\Wneg| = \sizewneg$ and $|\Wpos| = \sizewpos$. Let $\wnegoptim$ and $\wposoptim$ partition $\woptim$ in the same way $\Wneg$ and $\Wpos$ partition $W$. 

Because the optimal set of winners $\woptim$ depends from definition on the number of votes $\voptim$ that gives $i$ the highest utility, it may seem that $\wnegoptim$ and $\wposoptim$ are correlated. However, due to the non-symmetric nature of \cref{eq:v_optim}, we show that they are not.

\begin{lemma}
    $\wnegoptim$ is independent from $\wposoptim$ (i.e. only depends on $\sizewneg$ and $\sizewpos$).
\end{lemma}

\begin{proof} From \cref{eq:W_util}, if $\mathcal{C}$ represents a term that does not depend on $\Wneg$,
\small
\begin{equation*}
    \begin{split}
    U^i(\voptim, W, u) &= \frac{1}{\sizew}\sum_{\omega \in W} \util{i}{\omega} - \alpha\bigl[\sum_{\omega \in W}(\sumvotes{i}{\omega} - (\voptim + 1))^2 + \sum_{\omega \in \{1,\dots, m\}\setminus W} (\sumvotes{i}{\omega} - \voptim)^2 \bigr] \\
    &= \mathcal{C} + \bigl(\sum_{\omega \in \Wneg} \frac{\util{i}{\omega}}{\sizew} \bigr) -%
        \alpha\bigl[\sum_{\omega \in \{1,\dots,m\}} (\sumvotes{i}{\omega} - \voptim)^2 -%
        \sum_{\omega \in \Wneg}\bigl(2 (\sumvotes{i}{\omega}-\voptim)+1 \bigr)  \bigr] \\
     &= \mathcal{C} + \sum_{\omega \in \Wneg}\frac{\util{i}{\omega}}{\sizew} + \alpha \bigl( 2(\sumvotes{i}{\omega}-\voptim)+1 \bigr) \\
      &= \mathcal{C} + \sum_{\omega \in \Wneg} \underbrace{\frac{\util{i}{\omega}}{\sizew} + 2\alpha(\sumvotes{i}{\omega})}_{f(\omega)},
\end{split}
\end{equation*}
\normalsize
Notice that we were able to take $\voptim$ out as it does not depend on the contents of $\Wneg$.

Hence, $\wnegoptim$ contains the $x$ outcomes that maximize $f(\omega) = \frac{\util{i}{\omega}}{\sizew} + 2\alpha(\sumvotes{i}{\omega})$, for $\omega \in \{1,\dots,m\}$\footnote{Can be computed in $\O(w \log w)$ time by sorting $f(\omega)$ using Quicksort or any other efficient sorting algorithm.\label{sorttime}}.
\end{proof}

To compute $\woptim$, we also need to determine $\wposoptim$. Let us rearrange \cref{eq:W_util}.
\begin{equation}
     \begin{split}
     U^i(V, W, u)  &= \mathcal{T} + \sum_{\omega \in \Wpos} \underbrace{ \Bigl( \frac{\util{i}{\omega}}{\sizewneg+\sizewpos} - \alpha(V + 1 - \sumvotes{i}{\omega})^2 \Bigr)}_{g_\omega(V)} \\
     \end{split}
\end{equation}

Assume for a moment that $\voptim$ is fixed. Then, $\wposoptim$ contains the $\sizewpos$ candidates that satisfy $\sumvotes{i}{\omega} \leq \voptim$ and that have the highest value of outcome utility minus payment $g_\omega(V) = \frac{\util{i}{\omega}}{\sizewneg + \sizewpos} - \alpha(V + 1 - \sumvotes{i}{\omega})^2$, at $\voptim$. Due to the quadratic nature of the function, the plots of $g_{m+1}(V), g_{m+2}(V), \dots, g_{|\Omega|}(V)$ intersect at most $\O(\Omega^2)$ times, hence describing at most $\O(|\Omega|^2)$ possible orderings of $\{g_\omega(V)\}_{\omega=m+1}^{|\Omega|}$. The set of intersections is $G_I$ \eqref{eq:intersections} is
\begin{equation}
\begin{split}
    G_I &= \{V \mid \exists \omega, l \in \Omega \text{ such that } g_\omega(V) = g_\phi(V) \text{ and } \sumvotes{i}{\omega}, \sumvotes{i}{\phi} \leq V\} \cup \{-\infty, \infty\} \\
    &= \{V_1, V_2, \dots, V_{|G_I|}\}
\end{split}
\label{eq:intersections}
\end{equation}

Every interval $[V_j, V_{j+1}]$, where $V_j$ and $V_{j+1} \in G_I$, determines a particular order of the $g_\omega$ functions, as seen in \cref{fig:g_intersect}.

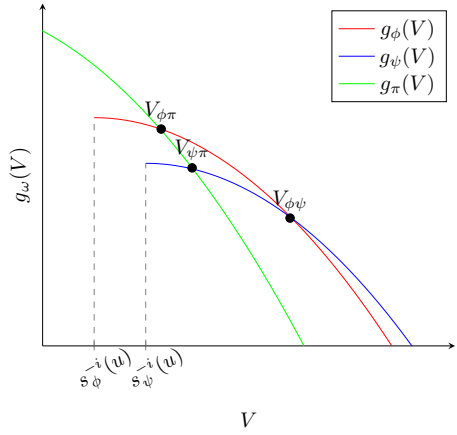
\begin{figure}[h!]
    \centering
    \begin{tikzpicture}[scale=0.8]
    \begin{axis}[
        axis lines = left,
        ytick=\empty,
        xlabel = \(V\),
        ylabel = {\(g_{\omega}(V)\)},
        xtick={5,6},
        xticklabels={$s^{-i}_{\phi}(u)$, $s^{-i}_{\psi}(u)$},
        x tick label style={rotate=30,anchor=north},
        ymax=100, xmin=4, ymin=-50, xmax=12
    ]
        \addplot [
            domain=5:20, 
            samples=100, 
            color=red,
        ] {50-3*(x-5)^2};    
        \addlegendentry{\(g_{\phi}(V)\)};
        
        \addplot [
            domain=6:20, 
            samples=100, 
            color=blue,
        ] {30-3*(x-6)^2};
        \addlegendentry{\(g_{\psi}(V)\)};
        
        \addplot [
            domain=3:20, 
            samples=100, 
            color=green,
        ] {100-3*(x-2)^2};
        \addlegendentry{\(g_{\pi}(V)\)};
    
        \draw[draw=gray, dashed, thin] (axis cs:6,-50) -- (axis cs:6,30);
        \draw[draw=gray, dashed, thin] (axis cs:5,-50) -- (axis cs:5,50);
        
        \addplot[mark=*] coordinates {(6.9,28)} node[above] {$V_{\psi\pi}$};
        \addplot[mark=*] coordinates {(6.3,45)} node[above] {$V_{\phi\pi}$};
        \addplot[mark=*] coordinates {(8.8,6)} node[above] {$V_{\phi\psi}$};
    \end{axis}
    \end{tikzpicture}
    
    \caption{The functions $g_\omega$ and the intersections described by $V \in G_I$. If $V \in [-\infty, V_{\phi\pi}]$, then $g_\pi\geq g_\phi \geq g_\psi$, however, if $V \in [V_{\phi\pi}, V_{\psi\pi}]$ then $g_\phi \geq g_\pi \geq g_\psi$, and so on. In this case, $G_I = \{-\infty, V_{\phi\pi}, V_{\psi\pi}, V_{\phi\psi}, \infty\}.$}
    \label{fig:g_intersect}

\end{figure}
 
Now, let $\voptim$ be variable again. Besides $\sumvotes{i}{m} > V \geq \sumvotes{i}{m+1}$, assume that $V_j \leq \voptim \leq V_{j+1}$ holds as well. With this extra assumption, it is very easy to obtain $\wjposoptim{j}$: pick, in linear time, the $|\Wpos| = x$ outcomes with the highest value of $g_\omega$ inside $[V_j, V_{j+1}]$. Since the intervals partition the whole $\R$ space for $V$, it must hold that $\wposoptim \in \{\wjposoptim{1}$, $\wjposoptim{2},\dots,\wjposoptim{|G_I|-1}\}$.

\newcommand{\vzoptim}[0]{V^{*,\Z}_m(V_j, V_{j+1},W)}
To restrict the domain of $\vmoptim$ back to $\Z$ (see \cref{V_domain}), we use that $U^i$ is concave in $V$ and, therefore, one of $ \lfloor \voptim \rfloor, \lceil \voptim \rceil$ is the integer maximizer. We must also account for the interval restrictions $\sumvotes{i}{m} > V \geq \sumvotes{i}{m + 1}$ and $V_j \leq V \leq V_{j+1}$; define $\vzoptim$ as such:
$$\vzoptim = \Bigl\{ V_{min}, \lfloor \vmoptim \rfloor, \lceil \vmoptim \rceil, V_{max} \Bigr\} \cap \Bigl[ V_{min},V_{max} \Bigr]$$

where $V_{min}=\max(\sumvotes{i}{m + 1}, V_j)$ and $V_{max}=\min (\sumvotes{i}{m}-1,V_{j+1})$.

Let $\mathscr{S}(m,\sizewneg, \sizewpos,j)$ to be the set of strategies that bring all outcomes in $W$ to $V + 1$ votes, $\Omega \setminus W$ to at most $V$ votes, where $W = \wmnegoptim \cup \wjmposoptim{j}$ and $V \in \vzoptim$. Then, define $\PBR$ to be the set of all such strategies across all $m,\sizewneg,\sizewpos,j$:

$$\PBR = \bigcup_{\sizewneg=1}^{m} \bigcup_{m=1}^{|\Omega|} \bigcup_{\sizewpos=1}^{|\Omega|-m} \bigcup_{j=0}^{|G_I(\sizewneg,\sizewpos)|} \mathscr{S}(m,\sizewneg,\sizewpos,j).$$

\Cref{alg:NE} iterates through $\PBR$ verifies if agent $i$ has a beneficial deviation by using any of the strategies in $\PBR$.

\begin{algorithm}
\caption{PTIME algorithm for $\DEVIATE$}
\begin{algorithmic}
\For{$\sizewneg \in \{0,\dots,|\Omega|\}$}
\For{$\sizewpos \in \{0,\dots,|\Omega|-\sizewneg\}$}
\State $G_I \gets \Call{intersectionSet}{\sizewneg,\sizewpos}$
\For {$m \in \{\sizewneg,\dots,|\Omega|-\sizewpos\}$}
\State $\Wneg \gets \wmnegoptim$
\For{$j \in \{1,\dots,|G_I|-1\}$}
\State $W \gets \Wneg \cup \wjmposoptim{j}$
\For{$V \in \vzoptim$}
\State $\sigma \gets \Call{getNewStrategy}{V, W}$
\If{$U^i(\allvotes{-i}, \sigma) > U^i(\allvotes{-i}, \allvotes{i})$} \Return Yes \EndIf
\EndFor \EndFor
\EndFor \EndFor
\EndFor
\Return No
\end{algorithmic}
\label{alg:NE}
\end{algorithm}

\completeness*
\begin{proof}[Proof]
\label{appendix:completeness}
Assume that agent $i$ has a beneficial deviation. It means that $\exists \sigma \in \Sigma$ such that $U^i(\allvotes{-i}, \sigma) > U^i(\allvotes{-i}, \allvotes{i})$. The strategy $\sigma$ corresponds to parameters $ w_{\sigma}, x_{\sigma}, m_{\sigma}, j_{\sigma}$. By construction, there is at least one strategy $\sigma^*$ in $\mathscr{S}( w_{\sigma}, x_{\sigma},m_{\sigma}, j_{\sigma})$ which returns the highest utility out of all possible strategies with the same parameters. Therefore, $U^i(\allvotes{-i}, \sigma^*) \geq U^i(\allvotes{-i}, \sigma) > U^i(\allvotes{-i}, \allvotes{i})$, where $\sigma^* \in \PBR$ as $\mathscr{S}(w_{\sigma}, x_{\sigma},m_{\sigma}, j_{\sigma}) \subseteq PBR$.
\end{proof}

From \cref{lemma:completeness} and \cref{lemma:soundness}, we deduce that \cref{alg:NE} is a decider for $\DEVIATE$. As $|G_I| \in \O(\Omega^2)$, $|\PBR| \in \O(\Omega^5)$, and all functions can be computed in polynomial time of $N$ and $\Omega$, it follows that $\DEVIATE \in$ P.
\end{proof}

\subsection{Section \ref{sec:complex_fixed_budget}}
\devb*
\begin{proof}
\label{appendix:dev_b}
Let $\lambda = \max_{\omega \in \Omega} (\sumvotes{i}{\omega})$ be the highest number of votes an outcome was cast without counting $i$'s votes. It must hold that $\lambda -\sqrt B \leq V \leq \lambda + \sqrt B$, since agent $i$ has a budget of $B$. Assume that $V$ is fixed for now.

\newcommand{\take}[1]{\text{take}({#1})}
\newcommand{\leave}[1]{\text{leave}({#1})}
We define $\take{\omega}$ to be the price to pay if agent $i$ wants to make outcome $\omega$ part of $W$ and $\leave{\omega}$ to be the price to pay if agent $i$ does not want $\omega$ to be part of $W$. For outcomes that satisfy $\sumvotes{i}{\omega} > V$, agent $i$ can bring them to $V + 1$ votes to make them part of $W$ or to $V$ otherwise. If $\sumvotes{i}{\omega} \leq V$, the agent has to pay only if $\omega \in W$.
\begin{gather}
\take{\omega} = ((V + 1) - \sumvotes{i}{\omega})^2, \forall \omega \in \Omega. \\
\leave{\omega} = 
\begin{cases}
    0, &\text{ if } \sumvotes{i}{\omega} \leq V \\
    (\sumvotes{i}{\omega} - V)^2, &\text{ otherwise.} \\
\end{cases}
\end{gather}

By including an outcome in $W$, the agent gains utility of $\util{i}{\omega}$ (divided by $|W|$, but we will include this later), otherwise none. Define $dp[n][p][b]$ to be the maximum utility sum agent $i$ can obtain by taking $p$ outcomes from the first $n$ by using a budget of at most $b$. Clearly, the definition makes sense for $p \leq n$. Denote $\{1,\dots,n\} \cap W$ by $W_n$.

\begin{equation*}
    dp[n][p][b] = \max_{W_n} \Bigl\{\sum_{\omega \in W_n} \util{i}{\omega} \hphantom{1} \Big| \hphantom{1} |W_n| = p \text{ and }
    \sum_{\omega \in W_n} \take{\omega} + \sum_{\omega \in \{1,\dots,n\} \setminus W_n} \leave{\omega} \leq b \Bigr\}
\end{equation*}

The initialization is $dp[0][0][b] = 0, \forall b \in \{0,\dots,B\}$. At each step, agent $i$ can choose to exclude element $n$ out of $W$ by obtaining utility $dp[n-1][p][b - \leave{n}]$ or to include $n$ into $W$ with $dp[n-1][p-1][b-\take{n}] + \util{i}{n}$, if the cases are well-defined (e.g. budget is not exceeded, etc).
The recurrence for $dp[n][p][b]$ is shown in \cref{eq:fixed_budget_dp0} and \cref{eq:fixed_budget_dp}, for all $n \in \{1,\dots, |\Omega|\}$, $b \in \{0,\dots,B\}$.

For $p = 0$:
\begin{equation}
\label{eq:fixed_budget_dp0}
dp[n][0][b] = 
\begin{cases}
    -\infty, &\text{ if } \leave{n} < b \\
    dp[n-1][0][b - \leave{n}], &\text{ otherwise. }  \\
\end{cases}
\end{equation}

For $p > 0$:
\begin{equation}
\label{eq:fixed_budget_dp}
dp[n][p][b] = 
\begin{cases}
    -\infty, &\text{ if } \leave{n}, \take{n} > b \\
    dp[n-1][p][b - \leave{n}], &\text{ if } \leave{n} \leq b < \take{n} \\
    dp[n-1][p-1][b-\take{n}] + \util{i}{n}, &\text{ if } \take{n} \leq b < \leave{n} \\

    \begin{aligned}
    \max \Bigl( &dp[n-1][p][b - \leave{n}], \\ 
                &dp[n-1][p-1][b-\take{n}] + \util{i}{n} \Bigr),
    \end{aligned} &\text{ otherwise. } 
\end{cases}
\end{equation}

The best response for $|W| = w$ gives agent $i$ utility $U^i(v^{-i}, V, w) = \frac{dp[|\Omega|][w][B]}{w}$. This method will be applied for all \\ $V \in \{\lambda-\sqrt B, \dots, \lambda+\sqrt B\}$ and the  pseudocode is presented in \cref{alg:fixed_budget_P}.

\begin{algorithm}
\caption{PTIME algorithm for $\DEVIATE_B$}
\label{alg:fixed_budget_P}
    \begin{algorithmic}
        \For {$V \in \{\lambda-\sqrt B, \dots, \lambda+\sqrt B\}$}
        \For {$n \in \{1,\dots,\Omega\}$}
        \For {$p \in \{0, \dots, n\}$}
        \For {$b \in \{0,\dots,B\}$}
        \State $\Call{Compute}{dp[n][p][b]}$
        \EndFor \EndFor \EndFor
        \For {$w \in \{1,\dots,\Omega\}$}
        \If {$U^i(\strtprofile) < dp[\Omega][w][B] / w$}
        \Return Yes
        \EndIf
        \EndFor
         \EndFor
        \Return No
    \end{algorithmic}%
\end{algorithm}%
The total time complexity is $\O(\Omega^2 B \sqrt{B})$, so $\DEVIATE_B \in P$.
\end{proof}

\end{document}